\documentclass{sig-alternate}

\usepackage{amsthm,amssymb,amsmath}
\usepackage{graphics}
\usepackage{tikz}
\usepackage[plainpages=false,pdfpagelabels,colorlinks=true,citecolor=blue,hypertexnames=false]{hyperref}
\usepackage{color}

\newcommand{\bN} { {\mathbb{N}}}
\newcommand{\bC} { {\mathbb{C}}}
\newcommand{\bQ} { {\mathbb{Q}}}
\newcommand{\bZ} { {\mathbb{Z}}}

\newcommand{\cont}{\operatorname{Cont}}
\newcommand{\lc}{\operatorname{lc}}
\newcommand{\rrem}{\operatorname{rrem}}

\newcommand{\si} { {\sigma}}

\newcommand{\pa}{\partial}

\newcommand{\ie}{{\it i.e.}}

\newtheorem{theorem}{Theorem}[section]
\newtheorem{lemma}[theorem]{Lemma}

\newtheorem{definition}[theorem]{Definition}

\newtheorem{prop}[theorem]{Proposition}
\newtheorem{corollary}[theorem]{Corollary}
\newtheorem{example}[theorem]{Example}

\newtheorem{algo}[theorem]{Algorithm}

\newfont{\mycrnotice}{ptmr8t at 7pt}
\newfont{\myconfname}{ptmri8t at 7pt}
\let\crnotice\mycrnotice%
\let\confname\myconfname%

\permission{Permission to make digital or hard copies of all or part of this work for personal or
classroom use is granted without fee provided that copies are not made or distributed for profit or
commercial advantage and that copies bear this notice and the full citation on the first page.
Copyrights for components of this work owned by others than ACM must be honored.
Abstracting with credit is permitted. To copy otherwise, or republish,
to post on servers or to redistribute to lists, requires prior specific permission and/or a fee.
Request permissions from permissions@acm.org.}

\conferenceinfo{ISSAC'16,}{July 20--22, 2016,  Waterloo, Canada.\\
{\mycrnotice{Copyright is held by the owner/author(s). Publication rights licensed to ACM.}}}
\copyrightetc{ACM \the\acmcopyr}
\crdata{978-1-4503-3435-8/15/07\ ...\$15.00.\\
Update with your DOI string http://XXXX provided in ACM eform confirmation email \& pdf of ACM eform}

\conferenceinfo{ISSAC'16,}{July 20--22, 2016,  Waterloo, Canada}
\copyrightetc{Copyright is held by the owner/author(s). Publication rights licensed to ACM.\\\the\acmcopyr}
\crdata{}
\doi{}

\CopyrightYear{2016}

\begin{document}

\allowdisplaybreaks

\title{Contraction of Ore Ideals with Applications}

\numberofauthors{1}

\author{%
 \alignauthor
 \leavevmode
 \mathstrut Yi Zhang\titlenote{Supported by the Austrian Science Fund (FWF) grants Y464-N18, NSFC grants (91118001, 60821002/F02)
 and a 973 project (2011CB302401).\vspace{-4pt}}\\[\smallskipamount]
  \affaddr{\leavevmode\mathstrut Institute for Algebra, Johannes Kepler University, Linz A-4040, Austria}\\
  \affaddr{\leavevmode\mathstrut KLMM, AMSS, Chinese Academy of Sciences, Beijing 100190, China}\\
  \affaddr{\leavevmode\mathstrut zhangy@amss.ac.cn}
}

\maketitle
\begin{abstract}
Ore operators form a common algebraic abstraction of linear ordinary differential and recurrence equations.
Given an Ore operator~$L$ with polynomial coefficients in~$x$, it generates a left ideal~$I$ in the Ore algebra
over the field~$\mathbf{k}(x)$ of rational functions. We present an algorithm for computing a basis of the contraction ideal of~$I$
in the Ore algebra over the ring~$R[x]$ of polynomials, where~$R$ may be either~$\mathbf{k}$ or a domain with~$\mathbf{k}$ as its fraction field.
This algorithm is based on recent work on desingularization for Ore operators by Chen, Jaroschek, Kauers and Singer.
Using a basis of the contraction ideal,
we compute a completely desingularized operator for~$L$ whose leading coefficient not only
has minimal degree in~$x$ but also has minimal content. Completely desingularized operators have interesting applications
such as certifying integer sequences and checking special cases of a conjecture of Krattenthaler.
\end{abstract}


\category{I.1.2}{Computing Methodologies}{Symbolic and Algebraic Manipulation}[Algorithms]


\terms{Algorithms, Theory}


\keywords{Ore Algebra, Desingularization, Contraction, Syzygy}

\section{Introduction}\label{SECT:intro}
There are various reasons why linear differential equations are
easier than non-linear ones. One is of course that the solutions
of linear differential equations form a vector space over the
underlying field of constants. Another important feature concerns
the singularities. While for a nonlinear differential equation
the location of the singularity may depend continuously on the
initial value, this is not possible for linear
equations. Instead, a solution $f$ of a differential equation
\[
  a_0(x)f(x) + \cdots + a_r(x)f^{(r)}(x) = 0,
\]
where $a_0,\dots,a_r$ are some analytic functions, can only have
singularities at points $\xi \in \bC$ with $a_r(\xi)=0$.

In this article, we consider the case where $a_0,\dots,a_r$ are
polynomials. In this case, $a_r$ can have only finitely many
roots.  We shall also consider the case of recurrence equations
\[
  a_0(n)f(n) + \cdots + a_r(n)f(n+r)=0,
\]
where again there is a strong connection between the roots of
$a_r$ and the singularities of a solution.

While every singularity of a solution leaves a trace in the
leading coefficient of an equation, the converse is not true. In
general, the leading coefficient $a_r$ may have roots at a point where no solution is singular. Such points are called
apparent singularities, and it is sometimes of interest to
identify them. One technique for doing so is called desingularization.
As an example, consider the recurrence operator
\[
 L = (1 + 16 n)^2 \pa^2 - (224 + 512 n) \pa - (1 + n)(17 + 16 n)^2,
\]
which is taken from~\cite[Section 4.1]{Abramov2006}.
Here, $\pa$ denotes the shift operator $f(n) \mapsto f(n + 1)$.
For any choice of two initial values
$u_0,u_1\in  \bQ$, there is a unique sequence $u \colon
\bN \to \bQ$ with $u(0)=u_0$, $u(1)=u_1$ and $L$ applied to $u$
gives the zero sequence. A priori, it is not obvious whether or
not $u$ is actually an integer sequence, if we choose $u_0,u_1$
from~$\bZ$, because the calculation of the $(n+2)$nd term
from the earlier terms via the recurrence encoded by $L$ requires
a division by $(1+16n)^2$, which could introduce fractions. In order
to show that this division never introduces a denominator, the
authors of~\cite{Abramov2006} note that every solution of $L$ is also a solution
of its left multiple
\begin{eqnarray*}
  T & = & \left( \frac{64}{(17 + 16 n)^2} \pa + \frac{(23 + 16 n)(25 + 16 n)}{(17 + 16 n)^2} \right) L \\
    & = & 64 \pa^3 + (16 n + 23) (16 n - 7) \pa^2 - (576 n + 928) \pa\\
    &   & - (16 n + 23) (16 n + 25)(n + 1).
\end{eqnarray*}
The operator $T$ has the interesting property that the factor
$(1+16n)^2$ has been ``removed'' from the leading
coefficient. This is, however, not quite enough to complete the
proof, because now a denominator could still arise from
the division by $64$ at each calculation of a new term via~$T$. To
complete the proof, the authors show that the potential denominators
introduced by $(1+16n)^2$ and by $64$, respectively, are in
conflict with each other, and therefore no such denominators can
occur at all.

The process of obtaining the operator $T$ from $L$ is called
desingularization, because there is a polynomial factor in the
leading coefficient of $L$ which does not appear in the leading
coefficient of~$T$. In the example above, the price to be paid
for the desingularization was a new constant factor $64$ which
appears in the leading coefficient of $T$ but not in the original
leading coefficient of~$L$. Desingularization algorithms in the
literature~\cite{Abramov1999, Abramov2006, Barkatou2015, Chen2013, Chen2016}
care only about the removal of polynomial
factors without introducing new polynomial factors, but they do not
consider the possible introduction of new constant factors.
A contribution of the present paper is a desingularization
algorithm which minimizes, in a sense, also any constant factors
introduced during the desingularization. For example, for the
operator $L$ above, our algorithm finds the alternative desingularization
\begin{equation} \label{EQ:ah}
\begin{array}{ccl}
\tilde{T}  & = & 1\pa^3 -\left(6272 n^3+3976 n^2+420 n+15\right) \pa^2 + \\
          &   & \left(12544 n^2+11871 n + 2782\right) \pa + 6272 n^4+ \\
          &   & 22792 n^3+30380 n^2 + 17459 n+3599,
\end{array}
\end{equation}
which immediately certifies the integrality of its solutions.

In more algebraic terms, we consider the following problem. Given
an operator $L\in \bZ[x][\pa]$, where $\bZ[x][\pa]$
is an Ore algebra (see Section~\ref{SECT:PRE} for definitions), we consider
the left ideal $\langle L \rangle = \bQ(x)[\pa]L$ generated by $L$ in the
extended algebra $\bQ(x)[\pa]$. The contraction of $\langle L \rangle$ to
$\bZ[x][\pa]$ is defined as $\cont(L) := \langle L \rangle \cap
\bZ[x][\pa]$. This is a left ideal of $\bZ[x][\pa]$ which
contains $\bZ[x][\pa]L$, but in general more operators.
Our goal is to compute a $\bZ[x][\pa]$-basis of $\cont(L)$.
In the example above, such a basis is given by $\{L, \tilde T\}$ (see Example~\ref{EX:ah}).
The traditional desingularization problem corresponds to computing
a basis of the $\bQ[x][\pa]$-left ideal $\langle L \rangle \cap \bQ[x][\pa]$.

The contraction problem for Ore algebras $\bQ[x][\pa]$ was proposed by Chyzak and Salvy~\cite[Section 4.3]{Salvy1998}.
For the analogous problem in
commutative polynomial rings, there is a standard solution via
Gr\"obner bases~\cite[Section 8.7]{Weispfenning1993}. It reduces the contraction problem to
a saturation problem. This reduction also works for the differential
case, but in that case it is not so helpful because it is less obvious
how to solve the saturation problem. A solution was proposed by
Tsai~\cite{Tsai2000}, which involves homological algebra and D-modules theory.
Our work is based on desingularization for Ore operators in~\cite{Chen2013,Chen2016}.
In particular, the $p$-removing operator in~\cite[Lemma~4]{Chen2016} provides us with a key
to determine contraction ideals. The algorithm developed in this paper is
considerably simpler than Tsai's and at the same time it applies to
arbitrary Ore algebras rather than only the differential case.
Moreover, we compute a completely desingularized operator in a contraction ideal, which
has minimal leading coefficient in terms of both degree and content.

The rest of this paper is organized as follows. In Section~\ref{SECT:PRE},
we describe Ore polynomial rings over a principal ideal domain, and extend the notion
of $p$-removed operators to them.  The notion of desingularized operators is defined and connected
with contraction ideals in Section~\ref{SECT:dc}. We determine a contraction ideal
in Section~\ref{SECT:algo}, and compute completely desingularized operators
in Section~\ref{SECT:cd}.

\section{Preliminaries}\label{SECT:PRE}

This section is divided into three parts. First, we describe Ore algebras that are used in the paper.
Second, we extend the notion of $p$-removed operators in~\cite{Chen2013,Chen2016}. At last, we
make some remarks  on Gr\"obner basis computation over a principal ideal domain.

\subsection{Ore Algebra}
Throughout the paper, we let $R$ be a principal ideal domain. For instance, $R$ can be
the ring of integers or that of univariate polynomials over a field.
We consider the \emph{Ore algebra}
$R[x][\pa; \sigma, \delta]$, where $\sigma: R[x] \rightarrow R[x]$ is a ring automorphism that leaves the
elements of $R$ fixed, and
$\delta: R[x] \rightarrow R[x]$ is a $\sigma$-derivation, \ie ~an $R$-linear map satisfying the skew Leibniz rule
$$\delta(fg) = \sigma(f)\delta(g) + \delta(f)g \quad \text{ for } f, g \in R[x].$$
The addition in $R[x][\pa]$ is coefficient-wise and the
multiplication is defined by associativity via the commutation rule~$\pa p = \sigma(p) \pa + \delta(p)$
for~$p \in R[x].$

Given $L \in R[x][\pa]$, we can uniquely write it as
\[
 L = \ell_r \pa^r + \ell_{r-1} \pa^{r-1} + \cdots + \ell_0
\]
with $\ell_0, \ldots, \ell_r \in R[x]$ and $\ell_r \neq 0$.
We call~$r$  the \emph{order} and~$\ell_r$ ~the \emph{leading coefficient} of $L$.
They are denoted by~$\deg_{\pa}(L)$ and~$\lc_{\pa}(L)$, respectively.
The ring~$R[x][\pa; \sigma, \delta]$ is abbreviated as~$R[x][\pa]$ when~$\sigma$ and~$\delta$ are clear from the context. For a subset~$S$ of~$R[x][\pa]$, the left ideal generated by~$S$ is denoted by~$R[x][\pa]\cdot S$.

Let~$Q_R$ be the quotient field of~$R$. Then~$Q_{R}(x)[\pa]$ is an Ore algebra containing~$R[x][\pa]$.
For~$L \in R[x][\pa]$,
we define the \emph{contraction ideal} of~$L$ to be~$Q_{R}(x)[\pa] L \cap R[x][\pa]$ and denote it by~$\cont(L)$.
\subsection{Removability} \label{SUBSECT:remove}
We generalize some terminologies given in~\cite{Chen2013, Chen2016} by replacing the coefficient ring~$\mathbb{K}[x]$ with~$R[x]$, where~$\mathbb{K}$ is a field.

\begin{definition}\label{DEF:premovable}
Let $L \in R[x][\pa]$ with positive order, and~$p$ be a divisor of~$\lc_{\pa}(L)$ in~$R[x]$.
\begin{itemize}
 \item[(i)] We say that $p$ is \emph{removable} from $L$ at
order $n$ if there exist~$P \in Q_{R}(x)[\pa]$ with order~$k$, and~$w, v \in R[x]$
with $\gcd(p, w) = 1$ in~$R[x]$ such that
$$PL \in R[x][\pa] \quad \text{and} \quad \sigma^{-k}(\lc_{\pa}(PL)) = \frac{w}{vp}\lc_{\pa}(L).$$
We call $P$ a \emph{$p$-removing operator for $L$ over~$R[x]$}, and $PL$ the corresponding \emph{$p$-removed operator}.
\item[(ii)] $p$ is simply called \emph{removable} from $L$ if it is removable at order $k$ for some $k \in \bN$. Otherwise,
$p$ is called \emph{non-removable} from $L$.
\end{itemize}
\end{definition}
Note that every $p$-removed operator lies in $\cont(L)$.

\begin{example}
In the example of Section~\ref{SECT:intro}, $(1 + 16n)^2$ is removable from $L$ at order 1.
And $T$ is the corresponding $(1 + 16n)^2$-removed operator for $L$.
\end{example}

\begin{example}
In the differential Ore algebra $\bZ[x][\pa]$, where $\pa x = x \pa + 1$,~let~$L = x (x - 1) \pa - 1$
Then $(1 {-} x) \pa^2 {-} 2 \pa {=} \left(\frac{1}{x}\pa \right) L$ is an $x$-removed operator for $L$
(see~\cite[Example 3]{Chen2013}).
\end{example}

The authors of~\cite{Chen2013} provide a convenient form of $p$-removing operators over~$\mathbb{K}[x]$ in order to get the order bound.
We derive a similar form over~$R[x]$ and use it in Section~\ref{SECT:cd}.

\begin{lemma}\label{premovable}
Let~$L \in R[x][\pa]$ with positive order. Assume that~$p \in R[x]$ is removable from~$L$ at order~$k$.
Then there exists a~$p$-removing operator for~$L$ over~$R[x]$ in the form
\[
  \frac{p_0}{\sigma^{k}(p)^{d_0}} + \frac{p_1}{\sigma^{k}(p)^{d_1}} \pa + \cdots
  + \frac{p_{k}}{\sigma^{k}(p)^{d_{k}}} \pa^{k},
\]
where~$p_i$ belongs to~$R[x]$, $\gcd(p_i, \sigma^{k}(p)) = 1$ in $R[x]$,~$i = 0,$ $1,$ \ldots, $k$, and~$d_k \geq 1$.
\end{lemma}
\begin{proof}
By~Definition~\ref{DEF:premovable},
$\lc_\pa(P) = \sigma^k \left(w/(vp)\right)$
for some~$w,v$ in~$R[x]$ with~$\gcd(w,p)=1$.  Then we can write a~$p$-removing operator for~$L$ over~$R[x]$ in the form
\[ P = \frac{p_0}{q_0 \sigma^{k}(p)^{d_0}} + \frac{p_1}{q_1 \sigma^{k}(p)^{d_1}} \pa + \cdots
  + \frac{p_{k}}{q_k \sigma^{k}(p)^{d_{k}}} \pa^{k},
\]
where~$p_i, q_i \in R[x]$,~$\gcd(p_i q_i, \sigma^{k}(p)) = 1$ in~$R[x]$,~$i = 0, \ldots, k$, $d_k \geq 1$.
Let~$\tilde{P} = \left( \prod_{i = 0}^k q_i \right) P$,~$\tilde{p_i} = p_i \left( \prod_{i = 0}^k q_i \right) / q_i$,
$i = 0, \ldots, k$. Then
\[ \tilde{P} = \frac{\tilde{p}_0}{\sigma^{k}(p)^{d_0}} + \frac{\tilde{p}_1}{\sigma^{k}(p)^{d_1}} \pa + \cdots
  + \frac{\tilde{p}_{k}}{\sigma^{k}(p)^{d_{k}}} \pa^{k},
\]
where~$\gcd(\tilde{p}_i, \sigma^{k}(p)) = 1$ in~$R[x]$,~$i = 0, \ldots, k$. Moreover,
$$\si^{-k}(\lc_\pa(\tilde{P} L)) = \frac{\si^{-k}(\tilde{p}_k)}{p^{d_k}} \lc_\pa(L).$$
By~Definition~\ref{DEF:premovable}, $\tilde{P}$ is a $p$-removing operator for~$L$ over~$R[x]$ with the required form.
\end{proof}
\subsection{Gr\"obner bases} \label{SUBSECT:gb}

In Sections~\ref{SECT:algo} and~\ref{SECT:cd}, we will make essential use of Gr\"obner bases in~$R[x][\pa]$.
When~$R=\mathbf{k}[t]$ with~$\mathbf{k}$ being a field, the notion of Gr\"obner bases and Buchberger's algorithm are available~\cite{Weispfenning1990}.
In our case, $\sigma$ is an $R$-automorphism of~$R[x]$, which implies that~$\sigma(x) = a x + b$
where~$a, b$ are in~$R$ and~$a$ is a unit.
Assume that~$\prec$ is a term order on~$\left\{x^i \pa^j \mid i, j \in \bN \right\}$. Let~$P$ be a nonzero operator in~$R[x][\pa]$,
and~$c$ be the head coefficient of~$P$ with respect to~$\prec$. By the commutation rule,~$\pa^i P$ has head coefficient~$c a^i$,
which is associated to~$c$, because~$a^i$ is a unit. This observation enables us to extend
the notion of Gr\"obner bases and Buchberger's algorithm in~\cite{Weispfenning1993, Saito1999} to Ore case in a straightforward way.
%
\section{Desingularization and \\ contraction}\label{SECT:dc}

In this section, we define the notion of desingularized operators, and connect it with contraction ideals.
As a matter of notation,
for an operator~$L \in R[x][\pa]$, we set
\[   M_k(L) = \left\{  P \in \cont(L) \mid  \deg_\pa(P) \le k \right\}. \]
Note that~$M_k(L)$ is a left submodule of~$\cont(L)$ over~$R[x]$. We call it the $k$th submodule
of~$\cont(L)$. When the operator~$L$ is clear from context, $M_k(L)$ is simply denoted by~$M_k$.
%
%
\begin{definition}\label{DEF:desingularization}
Let $L \in R[x][\pa]$ with order~$r > 0$, and
\begin{equation} \label{EQ:factor}
\lc_{\pa}(L) = c p_1^{e_1} \cdots p_m^{e_m},
\end{equation}
where~$c \in R$ and~$p_1,$ \ldots, $p_m \in R[x] \setminus R$ are irreducible and pairwise coprime.
An operator~$T \in R[x][\pa]$ of order~$k$ is called a {\em desingularized operator for $L$}
if $T \in \cont(L)$ and
\begin{equation} \label{EQ:dop}
\si^{r - k}(\lc_{\pa}(T)) = \frac{a}{b p_1^{k_1} \cdots p_m^{k_m}} \lc_{\pa}(L) ,
\end{equation}
 where $a, b \in R$, and~$p_i^{d_i}$ is non-removable from $L$ for each $d_i > k_i$, $i = 1 \ldots m$.
\end{definition}

Desingularized operators always exist by~\cite[Lemma 4]{Chen2016}.

\begin{lemma}\label{LM:key}
Let~$L \in R[x][\pa]$ be of~order~$r > 0$, and~$k \in \bN$ with~$k \geq r$.
Assume that~$T$ is a desingularized operator for~$L$ and $\deg_{\pa}(T) = k$.
\begin{itemize}
\item[(i)] $\deg_{x}(\lc_{\pa}(T)) = \min \{  \deg_{x}(\lc_{\pa}(Q)) \mid Q \in M_k(L) \setminus \{0\}\}.$
\item[(ii)] $\pa^i T$ is a desingularized operator for $L$
for each $i \in \bN$.
\item[(iii)] Set~$\lc_{\pa}(T) = a g$, where~$a \in R$ and~$g \in R[x]$ is primitive.
Then, for all~$F \in \cont(L)$ of order~$j$ with~$j \ge k$, $\si^{j - k}(g)$ divides $\lc_{\pa}(F)$ in~$R[x]$.
\end{itemize}
\end{lemma}
\begin{proof}
(i) Let $ t = \lc_{\pa}(T)$ and
$$d = \min \{\! \deg_{x}(\lc_{\pa}(Q))\! \mid Q \in M_k(L) \setminus \{0\} \}.$$
Suppose that~$d < \deg_{x}(t)$.
Then there exists $Q \in \cont(L)$ with~$\deg_{x}(\lc_{\pa}(Q)) = d$. Without loss of
generality, we can assume that~$\deg_{\pa}(Q) = k$, because the leading coefficients of~$Q$ and~$\pa^i Q$
are of the same degree for all~$i \in \bN$.

By pseudo-division in~$R[x]$, we have that
$$s t = q \lc_\pa(Q) + h$$
for some~$s {\in} R {\setminus} \{0\}$, $q, h {\in} R[x]$,
and~$h=0$ or~$\deg_x(h) < d$. If~$h$ were nonzero, then~$s T - q Q$ would be a nonzero operator of order~$k$ in~$\cont(L)$
whose leading coefficient is of degree less than~$d$, a contradiction. Thus,~$s t = q \lc_{\pa}(Q)$.
In particular,~$\deg_x(q)$ is positive, as~$d < \deg_{x}(t)$.
It follows from~\eqref{EQ:dop} that
$$\si^{r -k}(\lc_{\pa}(Q)) = \si^{r-k} \left( \frac{s t}{q}\right) = \frac{s a}{\sigma^{r-k}(q) b p_1^{k_1} \cdots p_m^{k_m}} \lc_{\pa}(L),$$
which belongs to~$R[x]$.  Hence,~$\si^{r-k}(q)$ divides~$\lc_{\pa}(L)$ in~$R[x]$. Consequently,~there exists $i \in \{1 \ldots m \}$ such that $p_i$ divides~$\si^{r-k}(q)$ in~$R[x]$. This implies that~$p^{k_i+1}$ is removable from $L$, a contradiction.

(ii) It is immediate from Definition~\ref{DEF:desingularization}.

(iii) Let~$\lc_{\pa}(F) = u f$, where~$u \in R$ and~$f$ is primitive in~$R[x]$.
By~(ii), $\pa^{j - k} T$ is a desingularized operator whose leading coefficient
equals~$a \si^{j - k}(g)$.
A similar argument used in the proof of the first assertion implies that
$$ v f =  p \si^{j - k}(g) \quad
\text{for some~$v \in R \setminus \{0\}$ and~$p \in R[x]$.}$$
By Gauss's Lemma in~$R[x]$ , $\si^{j - k}(g) \mid f$.
\end{proof}

We describe a relation between desingularized operators and contraction ideals.
Let~$I$ be a left ideal in~$R[x][\pa]$, and~$a \in R$. The saturation of~$I$ with respect to~$a$ is
defined to be
\[ I : a^\infty = \left\{ P \in R[x][\pa] \mid a^i P \in I \,\, \text{for some~$i \in \bN$} \right\}. \]
Since~$a$ is a constant with respect to~$\sigma$ and~$\delta$, the saturation~$I : a^\infty$ is a left ideal.

\begin{theorem}\label{TH:dc}
Let~$L \in R[x][\pa]$ with order $r>0$. Assume that~$T$ is a desingularized operator for~$L$.
Let~$\lc_{\pa}(T) {=} a g$, where $a \in R$ and $g$ is primitive in~$R[x]$.
If~$T$ belongs to~$M_k$ for some~$k \in \bN$,
then
$$\cont(L) = ( R[x][\pa] \cdot M_k ) : a^{\infty}.$$
\end{theorem}
\begin{proof} By Lemma~\ref{LM:key} (ii), we may assume that the order of~$T$ is equal to~$k$.
Set $J = ( R[x][\pa] \cdot M_k ) : a^{\infty}$.

First, assume that~$F \in J$. Then there exists~$j \in \bN$
such that $a^j F \in R[x][\pa] \cdot M_k$. It follows that~$F \in Q_{R}(x)[\pa] L$.
Thus, $F \in \cont(L)$ by definition.

Next, note that~$\cont(L) = \cup_{i = r}^{\infty} M_i$ and that~$M_i \subseteq M_{i + 1}$.
It suffices to show~$M_i \subseteq J \text{ for all } i \geq k$.
We proceed by induction on~$i$.

For~$i = k$. $M_k \subseteq J$ by definition.

Suppose that the claim holds for $i$. For any $F \in M_{i + 1} \backslash M_{i}$,
 $\deg_{\pa}(F) = i + 1$.
By Lemma~\ref{LM:key}~(iii), $\lc_{\pa}(F) = p \si^{i + 1 - k}(g)$
for some $p \in R[x]$. Then~$\lc_{\pa}(a F) = \lc_{\pa}(p \pa^{i + 1 - k} T)$.
It follows that~$a F - p \pa^{i + 1 - k} T \in M_{i}$.
Since
$$p \pa^{i + 1 - k} T \in R[x][\pa] \cdot M_k \subseteq R[x][\pa] \cdot M_i,$$
 we have that
$a F \in R[x][\pa] \cdot M_i$. On the other hand,~$M_i {\subset} J$ by the induction hypothesis.
Thus,~$a F \in R[x][\pa] \cdot J$, which is~$J$. Accordingly,~$F \in J$ by the definition
of saturation.
\end{proof}

\section{An Algorithm for computing \\ contraction ideals} \label{SECT:algo}
First, we translate an upper bound for the order of a desingularized operator over~$Q_R[x]$ to~$R[x]$.
\begin{lemma} \label{LM:bound}
Let~$L {\in} R[x][\pa]$ with order~$r>0$, and $p {\in} R[x]$ be a primitive polynomial and a divisor of~$\lc_{\pa}(L)$.
Assume that there exists a $p$-removing operator for~$L$ over~$Q_R[x]$.
Then there exists $p$-removing operator for~$L$ over~$R[x]$ with order~$r$.
\end{lemma}
\begin{proof}
Assume that~$P \in Q_R(x)[\pa]$ is a $p$-removing operator for~$L$ over~$Q_R[x]$.
Let~$P$ be of order~$k$.
Then~$PL$ is of the form
\[
 PL = \frac{a_{k + r}}{b_{k +r}} \pa^{k + r} + \cdots + \frac{a_1}{b_1} \pa + \frac{a_0}{b_0}
\]
for some~$a_i \in R[x], b_i \in R$, $i = 0, \ldots, k + r$. Moreover,
\[ \sigma^{-k}\left(\lc_{\pa}(PL)\right) = \frac{w}{vp} \lc_\pa(L), \]
where~$w,v \in R[x]$ with~$\gcd(w, p)=1$.

Let $b = \text{lcm}(b_0, b_1, \ldots, b_{k + r})$ in~$R$ and~$P' = b P$.
Then
\[ P'L \in R[x][\pa] \quad \text{and} \quad \si^{-k}(\lc_{\pa}(PL)) = \frac{bw}{vp} \lc_{\pa}(L).\]
Since $p$ is primitive, we have that~$\gcd(bw, p) = 1$ in~$R[x]$.
Thus, $P'$ is a $p$-removing operator of order $k$.
\end{proof}

By the above lemma, an order bound for a $p$-removing operator over~$Q_R[x]$ is also an order bound for
a $p$-removing operator over~$R[x]$. The former has been well-studied in the literature.
Order bounds for differential operators are given in~\cite[Algorithm 3.4]{Tsai2000} and~\cite[Lemma 4.3.12]{Max2013}.
Those for recurrence operators are given in~\cite[Lemma 4]{Chen2013} and~\cite[Lemma 4.3.3]{Max2013}.
Desingularized operators are $p$-removing operators. So we can find order bounds for them.

By Theorem~\ref{TH:dc}, determining a contraction ideal
amounts to finding a desingularized
operator~$T$ and an $R[x]$-basis of~$M_k$, where~$k$ is an upper bound for the order of~$T$.

Next, we present an algorithm for constructing a basis for~$M_k(L)$, where~$L$ is a nonzero operator in~$R[x][\pa]$
and~$k$ is a positive integer. To this end, we embed~$M_k$ into the free module~$R[x]^{k+1}$ over~$R[x]$.

Let us recall the right division in~$Q_{R}(x)[\pa]$~(see \cite[Section 3]{Bronstein1996}). For
$F, G \in Q_{R}(x)[\pa]$ with~$G \neq 0$, there exist unique elements $Q, R \in Q_{R}(x)[\pa]$ with~$\deg_{\pa}(R) < \deg_{\pa}(G)$ such that~$F = Q G + R$. We call~$R$ the \emph{right-remainder} of~$F$ by~$G$ and denote it by $\rrem(F, G)$.

Let~$F \in R[x][\pa]$ with order~$k$.
Then~$F \in M_k$ if and only if~$F \in Q_R[x][\pa] L$,
which is equivalent to~$\rrem(F, L)=0$. The latter gives rise to a linear system
\begin{equation} \label{EQ:linear}
(z_k,  \ldots, z_0) A  = \mathbf{0},
\end{equation}
where~$A$ is a $(k+1) \times r$ matrix over~$Q_R(x)$. Clearing denominators of the elements in~$A$, we may further assume
that~$A$ is a matrix over~$R[x]$. We are concerned with the solutions of~\eqref{EQ:linear} {\em over~$R[x]$}. Set
\[  N_k = \left\{(f_k,  \ldots, f_0) \in R[x]^{k+1} \mid (f_k,  \ldots, f_0) A = \mathbf{0} \right\}. \]
We call~$N_k$ the module of syzygies defined by~\eqref{EQ:linear}.
\begin{theorem} \label{TH:iso}
With the notation just specified, we have
\[
\begin{array}{cccc}
\phi: &  M_k & \longrightarrow & N_k \\
      &   \sum_{i=0}^k f_i \pa^i & \mapsto & (f_k, \ldots, f_0)
\end{array}
\]
is a module isomorphism over~$R[x]$.
\end{theorem}
\begin{proof}
Let~$F = \sum_{i=0}^k f_i \pa^i \in R[x]$.
If~$F$ belongs to~$M_k$, then~$\rrem(F, L)=0$, that is,~$(f_k, \ldots, f_0)$ belongs to~$N_k$.
Hence,~$\phi$ is a well-defined map.

Clearly,~$\phi$ is injective. For~$(f_k, \ldots, f_0) \in N_k$,
we have
$$(f_k, \ldots, f_0) A = \mathbf{0}.$$ As~\eqref{EQ:linear} is induced by right division~$\rrem\left( F, L \right)=0$,
$F$ belongs to~$M_k$. So~$\phi$ is surjective. It is straightforward to see that~$\phi$ is an $R[x]$-module homomorphism.
\end{proof}

By Theorem~\ref{TH:iso},~$M_k$ is finitely generated over~$R[x]$. To find an $R[x]$-basis of~$M_k$, it suffices to compute a basis of the module of syzygies defined by~\eqref{EQ:linear}.
When~$R$ is a field, we just need to solve~\eqref{EQ:linear} over a principal ideal domain~\cite[Chapter 5]{Arne2013}. When~$R$ is the ring of integers
or the ring of univariate polynomials over a field, we can use Gr\"obner bases of polynomials over a principal domain~\cite{Kapur1988, David182}. Their implementations are available in computer algebra systems such as {\tt Macaulay2} and {\tt Singular}.

We now consider how to construct a desingularized operator for~$L$. For~$k \in \bZ^+$, we define
\[  I_k = \left\{ [\pa^k] P \mid P \in M_k \right\} \cup \{0\}, \]
where~$[\pa^k] P$ stands for the coefficient of~$\pa^k$ in~$P$. It is clear that~$I_k$ is an ideal of~$R[x]$.
We call~$I_k$ the $k$th coefficient ideal of~$\cont(L)$. By the commutation rule,~$\sigma(I_k) \subset I_{k+1}$.
\begin{lemma}\label{LM:dop}
Let $L \in R[x][\pa]$ be of positive order.
If the $k$th submodule $M_k$ of~$\cont(L)$ has a basis~$\{B_1, \ldots, B_\ell \}$ over~$R[x]$,
then the $k$th coefficient ideal
$$I_k = \left\langle [ \pa^k ] B_1, \ldots,  [ \pa^k ] B_\ell  \right\rangle.$$
\end{lemma}
\begin{proof}
 Obviously,~$\langle [ \pa^k ] B_1, \ldots,  [ \pa^k ] B_\ell \rangle \subseteq I_k$.
 Let~$f \in I_k$. Then $f = \lc_{\pa}(F)$ for some~$F \in M_k$ with~$\deg_{\pa}(F) = k$.
 Since~$M_k$ is generated by~$\{B_1, \ldots, B_\ell \}$ over~$R[x]$,
\[
 F = h_1 B_1 + \cdots + h_\ell B_\ell, \quad
\text{where~$h_1, \ldots, h_\ell \in R[x]$.}\] Thus,
$f  = h_1 \left( [ \pa^k ] B_1 \right) + \cdots + h_\ell \left( [ \pa^k ] B_\ell \right).$
Consequently, $f \in \langle [ \pa^k ] B_1, \ldots,  [ \pa^k ] B_\ell \rangle$.
\end{proof}
\begin{theorem} \label{TH:dop}
Let~$L \in R[x][\pa]$ be of positive order. Assume that the $k$th submodule~$M_k$ of~$\cont(L)$
contains a desingularized operator for~$L$. Let~$s$ be a nonzero element in the $k$th coefficient ideal
with minimal degree. Then an operator~$S$ in~$M_k$ with leading coefficient~$s$ is a desingularized operator.
\end{theorem}
\begin{proof}
Let~$T$ be a desingularized operator in~$M_k$. By Lemma~\ref{LM:key}~(ii), we may assume that the order of~$T$ is equal to~$k$.
Let~$t = \lc_\pa(T)$. Then~$\deg(t)=\deg(s)$ by Lemma~\ref{LM:key}~(i).
Let~$u$ be the leading coefficient of~$s$ with respect to~$x$ and~$v$ be that of~$t$.
Then~$ut-vs$ is zero. Otherwise, $u T - v S$ would be an operator of order~$k$ whose leading coefficient with respect to~$\pa$
has degree lower than~$\deg_x(t)$, a contradiction to Lemma~\ref{LM:key}~(i).
It follows from $ut=vs$ and~Definition~\ref{DEF:desingularization} that~$S$ is also a desingularized operator.
\end{proof}

Let~$L$ be an operator in~$ R[x][\pa]$ of positive order. We can compute a basis~$\{B_1, \ldots, B_{\ell}\}$ for the~$k$th submodule of~$\cont(L)$ by Theorem~\ref{TH:iso}, where~$k$
is an upper bound on the order of a desingularized operator for~$L$. By Lemma~\ref{LM:dop}, we can obtain a basis~$\{b_1, \ldots, b_{\ell}\}$ for the $k$th coefficient ideal~$I_k$ of~$\cont(L)$. Let~$\bar{I}_k$ be the extension ideal of~$I_k$ in~$Q_R[x]$. Using the extended Euclidean algorithm in~$Q_R[x]$ and clearing denominators, we find
cofactors~$c_1,$ \ldots, $c_\ell \in R[x]$ and~$s \in R[x]$ such that
\[     \bar{I}_k = \langle s  \rangle \quad \text{and} \quad c_1 b_1 + \cdots + c_\ell b_\ell = s. \]
It follows from Theorem~\ref{TH:dop} that~$T=c_1 B_1 + \cdots + c_\ell B_\ell$ is a desingularized operator for~$L$ with~$\lc_\pa(T)=s$.
Let~$a$ be the content of~$s$ with respect to~$x$.
By Theorem~\ref{TH:dc},~$\cont(L)$ is the saturation of~$R[x][\pa] \cdot M_k$ with respect to~$a$.
Note that~$a$ belongs to~$R$, which is contained in the center of~$R[x][\pa]$. So a basis of the saturation ideal
can be computed in the same way as in the commutative case.




\begin{prop}\label{PROP:saconst}
 Let~$I$ be a left ideal of~$R[x][\pa]$ and~$c$ be non-zero element in~$R$. Assume that~$J$ is a left ideal~
$$R[x, y][\pa] \cdot \left( I \cup \{ 1 -c y \} \right),$$
where~$y$ is a new indeterminate and commutes with every element in~$R[x][\pa]$.
 Then~$I : c^{\infty}$ = $J \cap R[x][\pa]$.
\end{prop}
\begin{proof}
Since both~$y$ and~$c$ commute with~$\pa$, the argument in~\cite[page 266, Proposition 6.37]{Weispfenning1993}
carries over.
%
%
%
\end{proof}

We outline our method for determining contraction ideals.
\begin{algo}\label{ALGO:cont}
Given $L \in R[x][\pa]$, where $\pa x = (x + 1) \pa$ or $\pa x = x \pa + 1$,
compute a basis of $\cont(L)$.
\begin{enumerate}
 \item[(1)] Derive an upper bound~$k$ on the order of a desingularized operator for~$L$.
 \item[(2)] Compute an $R[x]$-basis of~$M_k$.
 \item[(3)] Compute a desingularized operator $T$, and set~$a$ to be the content of~$\lc_{\pa}(T)$ with respect to~$x$.
 \item[(4)] Compute a basis of $(R[x][\pa] \cdot M_k) : a^{\infty}$.
\end{enumerate}
\end{algo}
The termination of this algorithm is evident.
Its correctness follows from Theorem~\ref{TH:dc}.
We assume that
the commutation rule in~$R[x][\pa]$ is either~$\pa x = (x + 1) \pa$ or~$\pa x = x \pa + 1$ in~$R[x][\pa]$, because
we only know order bounds for those cases.
In step~1, the order bound is derived from~\cite[Lemma 4]{Chen2013} and~\cite[Algorithm 3.4]{Tsai2000}.
In step~2, we need to solve linear systems over~$R[x]$ as stated in Theorem~\ref{TH:iso}. This can be done
by Gr\"obner basis computation in a free $R[x]$-module of finite rank. In step~3,~$T$ is computed according to Theorem~\ref{TH:dop} and the extended Euclidean algorithm in~$Q_R[x]$. The last step is carried out according to Proposition~\ref{PROP:saconst}.



\begin{example}
In the shift Ore algebra~$\bQ[t][n][\pa]$, in which the commutation rule is~$\pa n = (n + 1) \pa$.
Consider
\[
 L = (n -1) (n + t) \pa + n + t + 1.
\]
By~\cite[Lemma 4]{Chen2013}, we obtain an order bound~$2$ for a desingularized operator. Thus,~~$M_2$ contains a desingularized operator for~$L$.
In step~2 of  Algorithm~\ref{ALGO:cont}, we find that~$M_2$ is generated by
\begin{eqnarray*}
 T_1 & = & (2 + t) n \pa^2 + (4 - n + t) \pa - 1, \\
 T_2 & = & (n - 1) n \pa^2 + 2 ( n - 1) \pa + 1,
\end{eqnarray*}
where~$T_1$ is a desingularized operator,~$\lc_{\pa}(T_1) = (2 + t) n$.
Using Gr\"{o}bner bases,~$\cont(L) = (\bQ[t][n][\pa] \cdot M_2 ) : (2 + t)^{\infty}$ is generated by $\{ L, T_2 \}$.
\end{example}

Let us consider the example in Section~\ref{SECT:intro}.
\begin{example} \label{EX:ah}
In the shift Ore algebra~$\bZ[n][\pa]$, let
\[
 L = (1 + 16 n)^2 \pa^2 - (224 + 512 n) \pa - (1 + n)(17 + 16 n)^2.
\]
By~\cite[Lemma 4]{Chen2013}, we obtain an order bound~$3$ for a desingularized operator. Thus,~~$M_3$ contains a desingularized operator for~$L$.
In step~2 of  Algorithm~\ref{ALGO:cont}, we find that~$M_3$ is generated by~$\{L, \tilde{T} \}$, where~$\tilde T$ is given in~\eqref{EQ:ah}.
Note that~$\lc_{\pa}(\tilde{T}) {=} 1$. Thus,~$\tilde T$ is a desingularized operator.
Consequently,
$$\cont(L) = (\bZ[n][\pa] \cdot \{L, \tilde{T} \}   ) : 1^{\infty} = \bZ[n][\pa] \cdot \{L, \tilde{T} \}.$$
\end{example}

\begin{example} \label{EX:bm}
In the differential Ore algebra~$\bZ[x][\pa]$, in which the commutation rule is~$\pa x = x \pa + 1$.
Consider the operator
$L = x \pa^2 - (x + 2) \pa + 2 \in \bZ[x][\pa]$ in~\cite{Barkatou2015}.
By~\cite[Algorithm 3.4]{Tsai2000}, we obtain an order bound~$4$ for a desingularized operator. Thus,~~$M_4$ contains a desingularized operator for~$L$.
In step~2 of  Algorithm~\ref{ALGO:cont}, we find that~$M_4$ is generated by~$\{L, \pa L, T \}$, where~$T = \pa^4 - \pa^3$.
Note that~$\lc_{\pa}(T) = 1$. Thus,~$T$ is a desingularized operator.
Consequently,
$$\cont(L) = (\bZ[x][\pa] \cdot \{L, \pa L, T \}   ) : 1^{\infty} = \bZ[x][\pa] \cdot \{L, T \}.$$
\end{example}

\section{Complete desingularization} \label{SECT:cd}
As seen in Section~\ref{SECT:intro}, the shift operator
\[ L = (1 + 16 n)^2 \pa^2 - (224 + 512 n) \pa - (1 + n)(17 + 16 n)^2 \]
has a desingularized operator~$T$ with leading coefficient~$64$.  %
But the content of~$\lc_\pa(L)$ is 1. The redundant content~$64$ has been removed by computing another desingularized operator~$\tilde{T}$ in~\eqref{EQ:ah}. This enables us to see immediately that the sequence annihilated by~$L$ is an
integer sequence when its initial values are integers.

Krattenthaler proposes a conjecture in~\cite{George2015}: Let~$(a_n)_{n \ge 0}$ and~$(b_n)_{n \ge 0}$ be two P-recursive sequences
over~$\bZ$ with leading coefficients~$n$. Then~$(n! a_n b_n)_{n \ge 0}$
is also a P-recursive sequence over~$\bZ$ with leading
coefficient~$n$. To test the conjecture for the two particular sequences, one may first compute an annihilator~$L$ of~$(n! a_n b_n)_{n \ge 0}$,
and then look for a nonzero operator in~$\cont(L)$ whose leading coefficient has \lq\lq minimal\rq\rq~content with respect to~$n$. When the content is equal to~$1$, the conjecture is true for these
sequences.

These two observations motivate us to define the notion of completely desingularized operators.
\begin{definition} \label{DEF:cd}
Let~$L \in R[x][\pa]$ with positive order, and~$Q$ a desingularized operator for~$L$.
Set~$\lc_\pa(Q)=c \, g$,  where~$c$ is the content of~$\lc_\pa(Q)$ with respect to~$x$ and~$g$
the corresponding primitive part. We call~$Q$ a {\em completely designularized operator} for~$L$
if~$c$ is a divisor of the content of the leading coefficient of every desingularized
operator for~$L$.
\end{definition}

To see the existence of completely designularized operators, we assume that~$L$ is of order~$r$.
For a desingularized operator~$T$ of order~$k$, equations~\eqref{EQ:factor} and~\eqref{EQ:dop} in Definition~\ref{DEF:desingularization} enable us to write
\begin{equation} \label{EQ:cp}
\sigma^{r - k}\left(\lc_\pa(T)\right) = c_T \, g,
\end{equation}
where~$c_T \in R$ and $g = p_1^{e_1-k_1} \cdots p_s^{e_m-k_m}.$ Note that~$g$ is primitive and independent of the choice
of desingularized operators.
\begin{lemma} \label{LM:ideal}
Let~$L \in R[x][\pa]$ with order~$r >0$.
Set~$I$ to be the set consisting of zero and~$c_T$ given in~\eqref{EQ:cp} for all desingularized
operators for~$L$. Then~$I$ is an ideal of~$R$.
\end{lemma}
\begin{proof}
By Definition~\ref{DEF:desingularization}, the product of a nonzero element of~$R$ and a desingularized operator for~$L$
is also a desingularized one. So it suffices to show that~$I$ is closed under addition.
Let~$T_1$ and~$T_2$ be two desingularized operators of orders~$k_1$ and~$k_2$, respectively. Assume that~$k_1 \ge k_2$.
By~\eqref{EQ:cp},
\[ \sigma^{r - k_1}\left(\lc_\pa(T_1)\right) = c_1 \, g \quad \text{and} \quad
\sigma^{r - k_2}\left(\lc_\pa(T_2)\right) = c_2 \, g. \]
If~$c_1+c_2=0$, then there is nothing to prove. Otherwise,
a direct calculation implies that
\[ \lc_\pa(T_1) = c_1 \si^{k_1-r}(g) \quad \text{and} \quad \lc_\pa\left(\pa^{k_1-k_2} T_2\right) = c_2 \si^{k_1-r}(g). \]
Thus,~$T_1+\pa^{k_1-k_2}T_2$ has leading coefficient~$(c_1+c_2) \si^{k_1-r}(g).$
Accordingly,~$T_1+\pa^{k_1-k_2}T_2$
is  a desingularized one, which implies that~$c_1 + c_2$ belongs to~$I$.
\end{proof}
Since~$R$ is a principal ideal domain,~$I$ in the above lemma is generated by an element~$c$, which corresponds to
a completely desingularized operator.

Let $\prec$ be a term order on~$\left\{x^i\pa^j \mid i, j \in \bN \right\}$. For any non-zero operator $P \in R[x][\pa]$, we define the {\em head term}
of $P$ to be the highest term appearing in~$P$ with respect to $\prec$, and denote it by~$\operatorname{HT}(P)$.

The next technical  lemma serves as a step-stone to construct completely desingularized operators.
\begin{lemma}\label{LM:sc}
Let $L \in R[x][\pa]$ with order~$r>0$, and~$k \ge r$. Then
$R[x][\pa] \cdot M_k {=} R[x][\pa] \cdot M_{k + 1}$ if and only if~$\si(I_k) {=} I_{k + 1}.$
\end{lemma}
\begin{proof}
%
Assume that~$\si(I_k) = I_{k + 1}$. Since $M_k \subset M_{k + 1}$, it suffices to prove
that~$M_{k + 1} \subset R[x][\pa] \cdot M_k$.

For each $T \in M_{k + 1} \setminus M_k$,  we have that~$\lc_{\pa}(T) \in \si(I_k)$.
Thus, there exists $F \in M_k$ such that $\si(\lc_{\pa}(F)) = \lc_{\pa}(T)$.
In other words,~$T - \pa F \in M_k$. Consequently, $T \in R[x][\pa] \cdot M_k$.

Conversely, assume that~$R[x][\pa] \cdot M_{k + 1}= R[x][\pa] \cdot M_k$.
It suffices to prove that $I_{k + 1} \subseteq \si(I_k)$ because $\si(I_k) \subseteq I_{k + 1}$ by definition.
Let $\mathcal{B}$ be an $R[x]$-basis of~$M_k$. Then $\mathcal{B}$ is also a basis of the left ideal~$R[x][\pa] \cdot M_k$.

Let $\prec$ be the term order such that
$x^{\ell_1} \pa^{m_1} {\prec} x^{\ell_2} \pa^{m_2}$ if either~$m_1 {<} m_2$ or~$m_1 {=} m_2$ and~$\ell_1 {<} \ell_2$.
Since~$\deg_{\pa}(P) {\leq} k$ for each $P \in \mathcal{B}$, S-polynomials and G-polynomials formed by elements in~$M_k$
have orders no more than~$k$~\cite[Definition 10.9]{Weispfenning1993}.  By Buchberger's algorithm, there exists a Gr\"{o}bner basis~$\mathcal{G}$ of~$R[x][\pa] \cdot \mathcal{B}$ with respect to $\prec$ such that~$\deg_{\pa}(G) \leq k$ for each $G \in \mathcal{G}$.

For~$p {\in} I_{k + 1} {\setminus} \{0\}$, there exists $T \in M_{k + 1} \setminus  M_k$ such that $\lc_{\pa}(T) {=} p$.
Since~$T {\in} R[x][\pa] \cdot M_{k + 1}$, we have~$T {\in} R[x][\pa] \cdot M_k$.  It follows that~$T$ is reduced to zero by~$\mathcal{G}$. Thus,
\begin{equation} \label{EQ:gb}
T = \sum_{G \in \mathcal{G}} V_G G \quad \text{ with~$\operatorname{HT}(V_G G) \preceq \operatorname{HT}(T)$}.
\end{equation}
By the choice of term order, $\deg_{\pa}(V_G G) \leq k + 1$. If~$V_G G$ is of order~$k+1$, then~$\lc_\pa(V_G G) = a_G  \, \si^{k  + 1 - d_G}(\lc_{\pa}(G))$,
where~$a_G$ is in~$R[x]$ and~$d_G$ is the order of~$G$. Comparing the leading coefficients of operators in both sides of~\eqref{EQ:gb} and noticing~$\deg_\pa(T)=k+1$, we have
\[ p = \sum_{\deg_{\pa}(V_G G) = k + 1} a_G \, \si^{k  + 1 - d_G}(\lc_{\pa}(G)). \]
It follows that
\begin{equation} \label{EQ:lc}
\si^{-1} \left(p\right) = \sum_{\deg_{\pa}(V_G G) = k + 1} \sigma^{-1}(a_G) \, \si^{k  - d_G}(\lc_{\pa}(G)).
\end{equation}
On the other hand,~$\si^{k - d_G}(\lc_{\pa}(G)) = \lc_\pa\left(\pa^{k-d_G} G\right)$ implies that~$\si^{k - d_G}(\lc_{\pa}(G)) \in I_k$.
We have that~$\si^{-1}(p) \in I_k$ by~\eqref{EQ:lc}. Thus,~$I_{k+1}  \subset \si(I_k)$.
\end{proof}
By the above lemma,~$I_j = \si^{j-\ell}(I_\ell)$ whenever~$j \ge \ell$ and~$\cont(L)=R[x][\pa]\cdot M_\ell$. In this case,
a basis of~$I_j$ can be obtain by shifting a basis of~$I_\ell$, which allows us to find a completely desingularized
operator.
\begin{theorem} \label{TH:gbcd}
Let $L \in R[x][\pa]$ with order~$r>0$.
Assume that the $\ell$th submodule~$M_\ell$ of~$\cont(L)$ contains a basis of~$\cont(L)$. Let~$I_\ell$ be the $\ell$th coefficient
ideal of~$\cont(L)$, and~$\mathbf{G}$ a reduced Gr\"obner basis of~$I_\ell$.
Let~$f \in \mathbf{G}$ be of the lowest degree in~$x$ and $F$ be the operator in~$\cont(L)$ with~$\lc_\pa(F)=f$.
Then $F$ is a completely desingularized operator for~$L$.
\end{theorem}
\begin{proof}
By Lemma~\ref{LM:ideal},~$\cont(L)$ contains a completely desingularized operator~$S$.
Let~$j=\deg_\pa(S)$. Then~$\lc_\pa(S)$ is in~$I_j$ for some~$j \ge \ell$.
By Lemma~\ref{LM:sc},~$\sigma^{j-\ell}(I_\ell)=I_j$. It follows that~$\si^{\ell - j}(\lc_\pa(S))$ belongs to $I_\ell$.
By~\eqref{EQ:cp}, we have
\[\sigma^{r - j}\left(\lc_\pa(S)\right) = c_S \, g, \]
where~$c_S \in R$ and~$g$ is a primitive polynomial in~$R[x]$.
A direct calculation implies that $\si^{\ell - j}(\lc_\pa(S)) = c_S \si^{\ell - r}(g)$.
Since~$\si^{\ell - j}(\lc_\pa(S)) \in I_\ell$,  so does~$c_S \si^{\ell - r}(g)$.

Note that~$F$ is a desingularized operator by Theorem~\ref{TH:dop}. By~\eqref{EQ:cp},~$\sigma^{r - \ell}\left(f\right) = c_F \, g,$
where~$c_F \in R$. Thus,~$f = c_F \si^{\ell - r}(g)$.

Since $\mathbf{G}$ is a reduced Gr\"obner basis of $I_\ell$, $f$ is the unique polynomial in~$\mathbf{G}$ with minimal degree.
Moreover,~$c_S \si^{l - r}(g)$ is of minimal degree in~$I_\ell$. So it can be reduced to zero by $f$. Thus,~$c_F \mid c_S$.
On the other hand,~$c_S \mid c_F$ by Definition~\ref{DEF:cd}. Thus,~$c_S$ and~$c_F$ are associated to each other. Consequently,~$F$ is a completely desingularized operator for~$L$.
\end{proof}

The construction in the above theorem leads to the following algorithm.
\begin{algo}\label{ALGO:cd}
Given $L \in R[x][\pa]$, where $\pa x = (x + 1) \pa$ or $\pa x = x \pa + 1$,
compute a completely desingularized operator for~$L$.
\begin{enumerate}
 \item[(1)] Compute a basis~$\mathcal{A}$ of~$\cont(L)$ by  Algorithm~\ref{ALGO:cont}.
 \item[(2)] Set $\ell$ to be the highest order of the elements in~$\mathcal{A}$.
            Compute an~$R[x]$-basis $\mathcal{B}$ of $M_\ell$.
 \item[(3)] Set~$\mathcal{B}^\prime = \{ B \in \mathcal{B} \mid \deg_\pa(B)=\ell\}$.
 Compute a reduced Gr\"obner basis $\mathbf{G}$ of
            $\left\langle \left\{ \lc_\pa(B) \mid B \in \mathcal{B}^\prime \right\}  \right\rangle.$
 \item[(4)] Set~$f$ to be the polynomial in~$\mathbf{G}$ whose degree is the lowest one in~$x$.
            Tracing back to the computation of step~3, one can find $u_B\in R[x]$ such that
            $f = \sum_{B \in \mathcal{B}^\prime} u_B \lc_\pa(B).$
 \item[(5)] Output $\sum_{B \in \mathcal{B}^\prime} u_B B$.
\end{enumerate}
\end{algo}
The termination of this algorithm is evident. Its correctness follows from Theorem~\ref{TH:gbcd}.
\begin{example}\label{Example2}
Consider two sequences~$(a_n)_{n \ge 0}$ and~$(b_n)_{n \ge 0}$ satisfying the following two recurrence equations~\cite{George2015}
\[
 n a_n  =  a_{n - 1} + a_{n - 2} \quad \text{and} \quad
 n b_n =   b_{n - 1} + b_{n - 5},
\]
respectively. The sequence~$c_n = n! a_n b_n$ has an annihilator $L \in \bZ[n][\pa]$  with
$$\deg_\pa(L)=10 \,\, \text{and} \,\, \lc_\pa(L) =  (n+10) (n^6+47 n^5+ \cdots + 211696 ).$$
In step~1 of the above algorithm, $\cont(L) = R[x][\pa]\cdot M_{14}$.
In steps~2 and 3, we observe that~$I_{14}$  is generated by~$n {+} 14 $.
In other words, we obtain a completely desingularized operator~$T$ of order~$14$
with~$\lc_\pa(T) = n + 14$. Translating into the recurrence equations of $c_n$, we arrive at
\[ n c_n = \alpha_1 c_{n - 1} + \cdots + \alpha_{14} c_{n - 14},  \]
where $\alpha_i \in \bZ[n]$, $i = 1, \ldots, 14$.
This verifies Krattenthaler's conjecture for the sequences~$a_n$ and~$b_n$.

Note that it is impossible to have a completely desingularized operator of order less than~$14$.
In fact, for some lower orders, one can obtain
\begin{eqnarray*}
\si^{-11}(I_{11}) & = &\langle 11104n, 4n(n - 466), n(n^2-34n+1336) \rangle, \\
\si^{-12}(I_{12}) & = &\langle 4n, n(n - 24) \rangle, \\
\si^{-13}(I_{13}) & = &\langle 2n, n(n - 26) \rangle.
\end{eqnarray*}
They cannot produce a leading coefficient whose degree and content are both minimal.
\end{example}

\begin{example}
Consider the following recurrence equations:
{\small \[ \begin{array}{lll}
n a_n  & = & (31 n - 6) a_{n - 1} + (49 n - 110) a_{n - 2} + (9 n - 225) a_{n - 3}, \\
n b_n  & = & (4 n + 13) b_{n - 1} + (69 n - 122) b_{n - 2} + (36 n - 67) b_{n - 3}.
\end{array} \]}
Let~$c_n = n! a_n b_n$,  which has an annihilator $L \in \bZ[n][\pa]$ of~order $10$ with~$\lc_\pa(L) {=} (n + 9) \alpha$,
where~$\alpha {\in} \bZ[n]$ and~$\deg_{n}(\alpha) {=} 20$.

By the known algorithms for desingularization in~\cite{Abramov1999, Abramov2006, Chen2013, Chen2016}, we find
that~$c_n$ satisfies the recurrence equation
\[ \beta n c_n = \beta_1 c_{n - 1} + \ldots + \beta_{10} c_{n - 10},
\]
where $\beta$ is an 853-digit integer, $\beta_i \in \bZ[n]$, $i = 1, \ldots, 10$.

On the other hand, Algorithm~\ref{ALGO:cd} finds a completely desingularized operator~$T$ for~$L$ of order~$14$ whose
leading coefficient is~$n + 14$. Translating into the recurrence equation of~$c_n$ yields
$n c_n = \gamma_1 c_{n - 1} + \cdots + \gamma_{14} c_{n - 14},$
where $\gamma_i \in \bZ[n]$.
\end{example}

Let~$L \in R[x][\pa]$ with positive order and~$T$ a desingularized operator for~$L$. Then
the degree of~$\lc_\pa(L)$ in~$x$ is equal to
\[  d:= \deg_x\left( \lc_\pa(L) \right) - (k_1 + \cdots + k_m), \]
where~$k_1, \ldots, k_m$ are given in Definition~\ref{DEF:desingularization}. Hence,~$\cont(L)$
cannot contain any  operator whose leading coefficient has degree lower than~$d$.

We provide a lower bound for the content of the leading coefficients of operators in~$\cont(L)$
with respect to the divisibility relation on~$R$.
To this end, we write
\[
 L =  a_k f_{k}(x) \pa^k + a_{k - 1} f_{k - 1}(x) \pa^{k - 1} + \cdots + a_0 f_{0}(x)
\]
where $a_i \in R$ and~$f_{i}(x) \in R[x]$ is primitive, $i = 0, 1, \ldots, k$.
We say that $L$ is \emph{$R$-primitive} if $\gcd(a_0, a_1, \ldots, a_k) = 1$.
Gauss's lemma in the commutative case  also holds for $R$-primitive polynomials.
\begin{lemma} \label{LM:Gauss}
Let~$P$ and $Q$ be two operators in~$R[x][\pa]$.
If~$P$ and~$Q$ are $R$-primitive, so is $P Q$.
\end{lemma}
\begin{proof} First, we recall a result in~\cite[Theorem 3.7, Corollary 3.8]{Johannes2011}.
Assume that~$A$ is a ring with endomorphism $\si : A \rightarrow A$ and $\si$-derivation $\delta : A \rightarrow A$.
Let~$I \subseteq A$ be a $\si$-$\delta$-ideal, that is, an ideal such that $\si(I) \subseteq I$ and $\delta(I) \subseteq I$. Then there exists a unique ring
homomorphism
\[
 \chi : A[\pa; \si, \delta] \rightarrow (A/I)[\tilde{\pa}; \tilde{\si}, \tilde{\delta}]
\]
such that $\chi |_{A}: A \rightarrow A/I$ is the canonical homomorphism, and~$\chi(\pa) = \tilde{\pa}$, where $\tilde{\si}$ and
$\tilde{\delta}$ are the homomorphism and $\tilde{\si}$-derivation induced by $\si$ and $\delta$, respectively.


Let $p$ be a prime element of~$R$ and $I = \langle p \rangle$ be the corresponding ideal in $R[x]$.
Then~$I$ is a $\si$-$\delta$-ideal. From the above paragraph, there exists a unique homomorphism
\[
 \chi : R[x][\pa; \si, \delta] \rightarrow (R[x]/I)[\tilde{\pa}; \tilde{\si}, \tilde{\delta}]
\]
such that $\chi |_{R[x]}: R[x] \rightarrow R[x]/I$ is the canonical homomorphism, and~$\chi(\pa) = \tilde{\pa}$.
Note that $\si^{-1}(I)\subset I$, because, for~$p f \in I$ with~$f \in R[x]$,
$\si^{-1}(p f) = p \si^{-1}(f) \in I$. It follows that~$\tilde{\si}$ is an injective endomorphism of $A/I$.
On the other hand, $R[x]/I$ is a domain because~$I$ is a prime ideal.
Thus, $(R[x]/I)[\tilde{\pa}; \tilde{\si}, \tilde{\delta}]$ is a domain because $R[x]/I$ is a domain and $\tilde{\si}$ is injective.
Since $P$ and $L$ are $R$-primitive, we have that~$\chi(P) \chi(L) \neq 0$.  So $\chi(P L) \neq 0 $, because~$\chi$ is a homomorphism.
Consequently, $P L$ is $R$-primitive.
\end{proof}

There are more sophisticated variants of Gauss's lemma for Ore operators in~\cite[Proposition 2]{Kovacic1972} and~\cite[Lemma 9.5]{Zhang2009}.

\begin{theorem}\label{TH:cremovable}
Let~$L \in R[x][\pa]$ with positive order and~$p$ be a non-unit element of $R$.
If~$L$ is $R$-primitive and~$p \mid \lc_{\pa}(L)$, then $p$ is non-removable.
\end{theorem}

\begin{proof}
Assume that~$p$ is removable, then there exists a $p$-removing operator~$P$ such that~$PL \in R[x][\pa]$.
By Lemma~\ref{premovable}, we can write
\[
 P = \frac{p_0}{p^{d_0}} + \frac{p_1}{p^{d_1}} \pa + \cdots +
     \frac{p_{k}}{p^{d_{k}}} \pa^{k}
\]
where~$p_i \in R[x]$, $\gcd(p_i, p) = 1$ in $R[x]$, $i = 0, \ldots, k$ and $d_k \geq 1$. Let
$d = \max_{0 \leq i \leq k} d_i$ and $P_1 = p^d P$. Then the content~$c$ of~$P_1$  with respect to~$\pa$ is~$\gcd(p_0, \ldots, p_k)$
because~$\gcd(p_i, p) = 1$, $i = 0, \ldots, k.$
Let~$P_1=cP_2.$ Then~$P_2$ is the primitive part of~$P_1$. In particular,~$P_2$ is $R$-primitive.
Then~$ c P_2 L = p^d PL$.
Since~$\gcd(c, p) = 1$ and~$PL \in R[x][\pa]$,~$p$ divides the content of~$P_2L$ with respect to~$\pa$.
Since~$p$ is a non-unit element of~$R$, $P_2L$ is not $R$-primitive, a contradition to~Lemma~\ref{LM:Gauss}.
\end{proof}

\begin{example}\label{Example1}
In the shift Ore algebra~$\bZ[n][\pa]$, consider a $\bZ$-primitive operator
\begin{eqnarray*}
 L & = & 3 (n+2) (3 n+4) (3 n+5) (7 n+3) \left(25 n^2+21 n+2\right) \\
   &   & \pa^2 + (-58975 n^6-347289 n^5-798121 n^4-902739 n^3 \\
   &   & -519976 n^2-141300 n -13680 ) \pa + 24 (2 n+1) \\
   &   & (4 n+1) (4 n+3) (7 n+10) \left(25 n^2+71 n+48\right),
\end{eqnarray*}
which annihilates~$ \binom{4 n}{n}+ 3^n$.
We observe that~$3$ is a constant factor of~$\lc_{\pa}(L)$. By Theorem~\ref{TH:cremovable},~$3$ is non-removable.
\end{example}

\section{Concluding  remarks}
In this paper, we determine a basis of a contraction ideal defined by an Ore operator in~$R[x][\pa]$,
and compute a completely desingularized operator whose leading coefficient is minimal in terms of both degree and content.
A more challenging topic is to consider the corresponding problems in the multivariate case.

Our algorithms rely heavily on the computation of Gr\"{o}bner bases over a principal ideal domain~$R$.
At present, the computation of Gr\"{o}bner bases
over~$R$ is not fully available in a computer algebra system. So the algorithms in this paper are not yet implemented.
To improve their efficiency, we need to use linear algebra over~$R$ as much as possible.


\section{Acknowledgement}
I am grateful to my supervisors Manuel Kauers and Ziming Li for initiating this research project, stimulating inspirational discussions
and helping me revise the paper.
I thank Shaoshi Chen for suggestions on the proof of Lemma~\ref{LM:Gauss} and valuable information on literatures.

\bibliographystyle{abbrv}

\begin{thebibliography}{10}

\bibitem{Abramov2006}
S.~A. Abramov, M.~Barkatou, and M.~van Hoeij.
\newblock Apparent singularities of linear difference equations with polynomial
  coefficients.
\newblock {\em AAECC}, 117--133, 2006.

\bibitem{Abramov1999}
S.~A. Abramov and M.~van Hoeij.
\newblock Desingularization of linear difference operators with polynomial
  coefficients.
\newblock In {\em Proc.\ of ISSAC'99}, 269--275, New York, NY,
  USA, 1999. ACM.

\bibitem{Barkatou2015}
M.~A. Barkatou and S.~S. Maddah.
\newblock Removing apparent singularities of systems of linear differential
  equations with rational function coefficients.
\newblock In {\em Proc.\ of ISSAC'15}, 53--60, New York, NY,
  USA, 2015. ACM.

\bibitem{Weispfenning1993}
T.~Becker and V.~Weispfenning.
\newblock {\em Gr\"{o}bner Bases, A Computational Approach to Commutative
  Algebra}.
\newblock Springer-Verlag, New York,
  USA, 1993.

\bibitem{Bronstein1996}
M.~Bronstein and M.~Petkov\v{s}ek.
\newblock An introduction to pseudo-linear algebra.
\newblock {\em Theoretical Computer Science}, pages 3 --33, 1996.

\bibitem{Chen2013}
S.~Chen, M.~Jaroschek, M.~Kauers, and M.~F. Singer.
\newblock Desingularization explains order-degree curves for {O}re operators.
\newblock In {\em Proc.\ of ISSAC'13}, 157--164, New York, NY,
  USA, 2013. ACM.

\bibitem{Chen2016}
S.~Chen, M.~Kauers, and M.~F. Singer.
\newblock Desingularization of {O}re operators.
\newblock {\em J. Symb.\ Comput.}, 74:617--626, 2016.

\bibitem{Zhang2009}
R.~C. Churchill and Y.~Zhang.
\newblock Irreducibility criteria for skew polynomials.
\newblock {\em Journal of Algebra}, 322:3797--3822, 2009.

\bibitem{Salvy1998}
F.~Chyzak and B.~Salvy.
\newblock Non-commutative elimination in {O}re algebras proves multivariate
  identities.
\newblock {\em J. Symb.\ Comput.\ }, 1998.


\bibitem{George2015}
G. Labahn et~al.
\newblock Workshop on symbolic combinatorics and computational differential
  algebra.
\newblock {{\tt http://www.fields.utoronto.ca/video-archive/event/411/2015}}.

\bibitem{David182}
D.~Grayson, M.~Stillman, and D.~Eisenbud.
\newblock Macaulay2doc -- macaulay2 documentation.

\bibitem{Max2013}
M.~Jaroschek.
\newblock {\em Removable Singularities of Ore Operators}.
\newblock PhD thesis, RISC-Linz, Johannes Kepler Univ., 2013.

\bibitem{Kapur1988}
A.~Kandri-Rody and D.~Kapur.
\newblock Computing a {G}r\"{o}bner basis of a polynomial ideal over a {E}uclidean
  domain.
\newblock {\em J. Symb.\ Comput.\ }, pages 37--57, 1988.

\bibitem{Weispfenning1990}
A.~Kandri-Rody and V.~Weispfenning.
\newblock Non-commutative {G}r\"obner bases in algebras of solvable type.
\newblock {\em J.\ Symb.\ Comput.\ }, pages 1--26, 1990.

\bibitem{Christoph2009}
C.~Koutschan.
\newblock {\em Advanced Applications of the Holonomic Systems Approach}.
\newblock PhD thesis, Johannes Kepler University Linz, 2009.

\bibitem{Christoph2010}
C.~Koutschan.
\newblock {\em Holonomic{F}unctions ({U}ser's {G}uide)}.
\newblock RISC Report Series, Johannes Kepler Univ., 2010.

\bibitem{Kovacic1972}
J.~Kovacic.
\newblock An {E}isenstein criterion for noncommutative polynomials.
\newblock {\em Proceedings of the American Mathematical Society}, 34, 1972.

\bibitem{Johannes2011}
J.~Middeke.
\newblock {\em A computational view on normal forms of matrices of {O}re
  polynomials}.
\newblock PhD thesis, Johannes Kepler University Linz, 2011.

\bibitem{Saito1999}
M.~Saito, B.~Sturmfels, and N.~Takayama.
\newblock {\em Gr\"{o}bner Deformations of Hypergeometric Differential Equations.}
\newblock Springer-Verlag, New York, USA, 1999.

\bibitem{Arne2013}
A.~Storjohann.
\newblock {\em Algorithms for Matrix Canonical Forms}.
\newblock PhD thesis, Swiss Federal Institute of Technology Zurich, 2013.

\bibitem{Tsai2000}
H.~Tsai.
\newblock Weyl closure of a linear differential operator.
\newblock {\em J.\ Symb.\ Comput.\ }, pages 747--775, 2000.

\end{thebibliography}
\def\cprime{$'$}

\end{document}